\documentclass[letter, 10pt, conference]{ieeeconf}

\IEEEoverridecommandlockouts
\usepackage{amsmath,amssymb,amsfonts}
\usepackage{bm}
\usepackage{graphicx}
\usepackage{subcaption}

\usepackage{xcolor}
\usepackage[hidelinks]{hyperref}
\usepackage[ruled,vlined,linesnumbered]{algorithm2e}
\usepackage{cite}
\let\labelindent\relax

\usepackage{enumitem}

\setenumerate{label=\roman*), ref=(\roman*)}

\usepackage{amsthm}
\theoremstyle{plain}
\newtheorem{theorem}{Theorem}
\newtheorem{proposition}{Proposition}
\newtheorem{corollary}{Corollary}
\newtheorem{assumption}{Assumption}
\theoremstyle{remark}

\newcommand{\R}{\mathbb{R}}
\newcommand{\one}{\mathbf{1}}
\newcommand{\Vcal}{\mathcal{V}}
\newcommand{\Ecal}{\mathcal{E}}
\newcommand{\Gcal}{\mathcal{G}}
\newcommand{\Ncal}{\mathcal{N}}
\newcommand{\Hcal}{\mathcal{H}}
\newcommand{\Xcal}{\mathcal{X}}
\newcommand{\sat}{\mathrm{sat}}
\let\oldfrac\frac 
\renewcommand{\frac}[2]{%
  \mathchoice
    {\oldfrac{#1}{#2}}%
    {{#1}/{#2}}%
    {{#1}/{#2}}%
    {{#1}/{#2}}%
}
\title{\LARGE \bf Recursive Neighbor Pruning: Decentralized Community Detection\\ with Certified Behavior from Saturated Opinion Dynamics}

\author{Anthony Couthures$^{1}$, Irinel-Constantin Mor\u{a}rescu$^{2}$, Vineeth Satheeskumar Varma$^{2}$, \\ Samson Lasaulce$^{2}$ and Antoine Girard$^{3}$%
\thanks{$^{1}$UCLouvain, ICTEAM Institute, Louvain-la-Neuve, 1348, Belgium.}%
\thanks{$^{2}$Universit\'e de Lorraine, CNRS, CRAN, F-54000 Nancy, France.}%
\thanks{$^{3}$Universit\'e Paris-Saclay, CNRS, CentraleSup\'elec, Laboratoire des signaux et syst\`emes, F-91190, Gif-sur-Yvette, France.}%
\thanks{An extended version of this work, which also covers two variants of RNP, has been submitted to Physical Review E. The document has been checked with Claude Opus 5.5 in order to improve to overall presentation}%
}

\begin{document}
\bstctlcite{IEEEexample:BSTcontrol}

\maketitle
\thispagestyle{empty}
\pagestyle{empty}

\begin{abstract}
Detecting communities, i.e., groups of densely connected nodes, usually requires a central entity that knows the whole network. We study Recursive Neighbor Pruning (RNP), a fully decentralized alternative inspired by social learning. Each agent holds a private opinion but only shows its neighbors a coarse public action, namely its opinion amplified and clipped to a fixed range. Agents repeatedly move their opinion toward the actions of their neighbors and, once opinions settle, drop the links with neighbors whose action has the opposite sign. Repeating this procedure on the remaining network splits it into communities, without any agent knowing the network or the number of communities. We use tools from nonlinear control to characterize this behavior. We show that splitting requires the amplification to exceed a threshold set by how strongly the network is bottlenecked. We also show that a group in which every member has enough of its links inside the group can no longer be split once its members agree, and that a single run separates such groups exactly when started from suitable opinions. Finally, we bound the number of runs and messages. On two standard families of synthetic networks, RNP is as accurate as the best centralized method as long as each node has fewer than about forty percent of its links outside its community, and it fails close to the value predicted by our analysis. It is more accurate than two decentralized competitors below this level, but exchanges many more messages.
\end{abstract}

\section{Introduction}\label{sec:intro}

Many networks are organized into communities, i.e., groups of nodes with many links among themselves and few links to the rest of the network. Communities shape how such networks function and how information spreads in them~\cite{fortunatoCommunityDetectionGraphs2010,liComprehensiveReviewCommunity2024}. Standard detection methods, based on spectral clustering~\cite{pothenPartitioningSparseMatrices1990,spielmanSpectralPartitioningWorks2007} or on modularity maximization~\cite{blondelFastUnfoldingCommunities2008,traagLouvainLeidenGuaranteeing2019}, require a single entity that knows every link. In many networked systems this is impossible, because the network is too large or because nodes cannot or will not share their connections. This calls for algorithms in which communities emerge from local interactions between neighbors, and whose behavior can be guaranteed.

Social learning offers a template. People do not know the whole social network, nor what their acquaintances privately think, yet they infer which group they belong to by observing whether their neighbors' visible choices agree with theirs. Opinion dynamics models~\cite{frenchFormalTheorySocial1956,degrootReachingConsensus1974,friedkinSocialInfluenceOpinions1990} formalize this process, but in their basic form, where each agent averages its neighbors' opinions, all agents of a connected network end up agreeing~\cite{bergerNecessarySufficientCondition1981}, which erases any trace of communities. Persistent disagreement requires an additional mechanism, such as ignoring neighbors whose opinions are too different~\cite{deffuantMixingBeliefsInteracting2000,hegselmannOpinionDynamicsBounded2002}, antagonistic links~\cite{altafiniConsensusProblemsNetworks2013}, or coarse communication~\cite{martinsContinuousOpinionsDiscrete2008,grayMultiagentDecisionMakingDynamics2018,bizyaevaNonlinearOpinionDynamics2023}. In the last case, each agent communicates an amplified and clipped version of its opinion, like a hesitant person who votes firmly. When the amplification is strong enough, the neutral agreement becomes unstable and groups of agents polarize toward opposite extremes~\cite{fontanMultiequilibriaAnalysisClass2018,leonardNonlinearFeedbackDynamics2021a,couthuresGlobalSynchronizationMultiAgent2025}.

Recursive Neighbor Pruning (RNP) turns this polarization into a community detection algorithm. All agents start from random opinions and run the dynamics until opinions settle; each agent then drops its links to neighbors whose public action has the opposite sign; the procedure is repeated on the remaining network until no more links are dropped. The communities are the groups of nodes that remain connected. Agents never share their private opinion exactly, nor any community label, and they need no knowledge of the network or of the number of communities.

\emph{Related work.} Three lines of work are closest to ours. Mor\u{a}rescu and Girard~\cite{morarescuOpinionDynamicsDecaying2011} detect communities as groups of agents that agree quickly, by cutting links whose disagreement exceeds a shrinking tolerance while the dynamics evolve. RNP instead keeps the network fixed during each run and exploits the several stable polarized states created by clipping. Xing et al.~\cite{xingLearningCommunitiesEquilibria2024} show that the steady states of nonlinear opinion dynamics reveal two planted communities, which are then read by a central observer. RNP is fully decentralized and handles an unknown number of communities. Finally, existing decentralized methods either let nodes broadcast community labels (label propagation~\cite{raghavanLinearTimeAlgorithm2007,xieSLPAUncoveringOverlapping2011}) or compute spectral quantities by repeated local averaging~\cite{kempeDecentralizedAlgorithmSpectral2008,linPowerIterationClustering2010}; we compare RNP with both. This paper builds on~\cite{couthuresConsensusRobustClustering2025c} but we propose the following specific novelties.

\emph{Contributions.} (i) For the step-by-step dynamics actually run by the agents, we show that no link can be cut unless the amplification exceeds a threshold determined by the most pronounced bottleneck of the network. We also show that splitting starts along the direction used by spectral clustering. (ii) We give an explicit condition on the fraction of a node's links that leave its group. Under this condition, a group that has agreed can never be split. We also prove that a single run separates such groups exactly from an explicit set of starting opinions. (iii) We bound the number of runs and messages of RNP. (iv) We compare RNP with centralized and decentralized methods on two families of synthetic benchmark networks. The numerical simulations validate our theoretical results.

\emph{Notation.} $\R$ is the set of real numbers. Vectors and matrices are in bold, and $\bm{x}^\top$ is the transpose of $\bm{x}$. $\one$ is the vector with all entries equal to $1$, $\bm{0}$ is the zero vector and $\bm{I}$ the identity matrix; their dimension is the number $N$ of agents defined in Section~\ref{sec:model}. $\mathrm{diag}(\bm{x})$ is the diagonal matrix with diagonal $\bm{x}$, $\mathrm{span}(\bm{v})=\{\alpha\bm{v}\mid\alpha\in\R\}$, $|\mathcal{S}|$ is the number of elements of a set $\mathcal{S}$, $\mathrm{Fix}(f)$ is the set of fixed points of a map $f$, $\lceil\cdot\rceil$ rounds up to an integer, $\mathcal{U}(a,b)$ is the uniform distribution on $[a,b]$, and $\sat(y)=\min(1,\max(-1,y))$.

\section{Problem Setting and Model}\label{sec:model}

\subsection{Graph and decentralized constraints}

We consider $N$ agents interacting over an undirected, connected, simple graph $\Gcal=(\Vcal,\Ecal)$, $\Vcal=\{1,\dots,N\}$, with adjacency matrix $\bm{A}$, neighborhoods $\Ncal_i=\{j\in\Vcal\mid(i,j)\in\Ecal\}$, degrees $d_i=|\Ncal_i|\ge 1$, and $\bm{D}=\mathrm{diag}(d_1,\dots,d_N)$. The random-walk matrix $\bm{P}=\bm{D}^{-1}\bm{A}$ is self-adjoint for the inner product $\langle\bm{x},\bm{y}\rangle_{\bm{D}}=\bm{x}^\top\bm{D}\bm{y}$, with induced norm $\|\bm{x}\|_{\bm{D}}=\langle\bm{x},\bm{x}\rangle_{\bm{D}}^{1/2}$; vectors orthogonal for this inner product are called $\bm{D}$-orthogonal. Hence the eigenvalues of $\bm{P}$ are real, $-1\le\lambda_1\le\dots\le\lambda_{N-1}<\lambda_N=1$, and the associated eigenvectors $\bm{v}_1,\dots,\bm{v}_N$ can be chosen $\bm{D}$-orthogonal, with $\bm{v}_N=\one$ since $\bm{P}\one=\one$. The eigenvector $\bm{v}_{N-1}$ is the Fiedler vector of the random-walk Laplacian $\bm{I}-\bm{P}$. We write $\rho(\Gcal) = \max(\lambda_{N-1},-\lambda_1)$ and $K^*(\Gcal)=1/\lambda_{N-1}$ with convention $K^*(\Gcal)=+\infty$ if $\lambda_{N-1}\le 0 $ . 

We seek a partition of $\Vcal$ into communities under three constraints: \emph{locality} (agent $i$ only interacts with $\Ncal_i$ and knows no global quantity), \emph{limited disclosure} (agents exchange a single scalar per step, and never labels, degrees or neighbor lists), and \emph{uniformity} (identical rules, no leader or central aggregation).

\subsection{Nonlinear social learning dynamics}

At each time step $t=0,1,2,\dots$, each agent $i$ holds a private opinion $x_i(t)\in[-1,1]$ and broadcasts the public action $s(x_i(t))$, where $s:[-1,1]\to[-1,1]$ is the action function specified in Assumption~\ref{ass:s} below. It moves toward the average action of its neighbors:
\begin{equation}\label{eq:dyn_i}
x_i(t+1)=x_i(t)+h\Big(\frac{1}{d_i}\sum_{j\in\Ncal_i}s(x_j(t))-x_i(t)\Big),
\end{equation}
with step size $h\in(0,1]$. With $\bm{x}=(x_1,\dots,x_N)^\top$ and $\bm{s}(\bm{x})=(s(x_1),\dots,s(x_N))^\top$, this reads \[\bm{x}(t+1)=(1-h)\bm{x}(t)+h\bm{P}\bm{s}(\bm{x}(t)).\] This is the Euler discretization of $\dot{\bm{x}}=\bm{P}\bm{s}(\bm{x})-\bm{x}$ studied in~\cite{couthuresGlobalSynchronizationMultiAgent2025,couthuresConsensusRobustClustering2025c}. As $x_i(t+1)$ is a convex combination of $x_i(t)$ and values of $s$, $\Xcal=[-1,1]^N$ is forward invariant.

\begin{assumption}\label{ass:s}
   The function $s:[-1,1]\to[-1,1]$ satisfies $0 \leq (s(a)-s(b))/(a-b) \leq K$ for all $a \neq b$, is differentiable at $0$ with $s'(0)=K>1$, and $\mathrm{Fix}(s)=\{-1,0,1\}$.
\end{assumption}

The natural example, used in our experiments, is $s_K(x)= \min(\max(Kx,-1),1) :=\sat(Kx)$, which is odd. Monotonicity reflects cooperative behavior and makes~\eqref{eq:dyn_i} a monotone system~\cite{hirschChapter4Monotone2006}. Note that a neighbor observing $s_K(x_i)$ can reconstruct $x_i$ exactly when $|x_i|<1/K$, but only its sign once $|x_i|\ge 1/K$. This is not a formal privacy guarantee: the sign is always revealed, and repeated observations may reveal more. Assuming $K>1$ makes consensus at $\bm{0}$ repelling and $\pm\one$ attracting, creating a symmetry breaking.

\begin{algorithm}[ht]
    \DontPrintSemicolon
    \SetKwInOut{Input}{Input}\SetKwInOut{Output}{Output}
    \Input{$\Gcal=(\Vcal,\Ecal)$; $s$, $h$, $T_{\max}$, $N_{\mathrm{retry}}$}
    \Output{Pruned graph, communities as connected components}
    $n\gets 0$\;
    \While{$n<N_{\mathrm{retry}}$}{
      \lForEach{$i\in\Vcal$}{$x_i\gets\mathcal{U}(-1,1)$}
      \For{$t=1$ \KwTo $T_{\max}$}{
        \lForEach{$i\in\Vcal$}{$x_i\gets x_i+h\big(\oldfrac{1}{d_i}\sum_{j\in\Ncal_i}s(x_j)-x_i\big)$}
      }
      $\Ecal_{\mathrm{cut}}\gets\{(i,j)\in\Ecal\mid s(x_i)s(x_j)<0\}$\;
      \lIf{$\Ecal_{\mathrm{cut}}\neq\emptyset$}{$\Ecal\gets\Ecal\setminus\Ecal_{\mathrm{cut}}$; update $\Ncal_i,d_i$; $n\gets0$}
      \lElse{$n\gets n+1$}
    }
    \caption{Recursive Neighbor Pruning (RNP)}
    \label{algo:RNP}
\end{algorithm}\vspace{-5mm}

\section{The Recursive Neighbor Pruning Algorithm}\label{sec:rnp}

A \emph{round} of RNP (Algorithm~\ref{algo:RNP}) consists of: 
\begin{enumerate}[wide, labelwidth=!, labelindent=0pt,nosep,label=\arabic*.]
    \item independent draws $x_i(0)\sim\mathcal{U}(-1,1)$
    \item $T_{\max}$ steps of dynamics~\eqref{eq:dyn_i} on the current graph
    \item removal of every edge $(i,j)$ with $s(x_i)s(x_j)<0$, a local action since agent $i$ knows its action and observes that of $j$.
\end{enumerate}
Then, if an edge is cut, a new round starts on the new graph; otherwise the round is repeated, up to $N_{\mathrm{retry}}$ consecutive unsuccessful rounds. The communities are the connected components of the final graph.

In Fig.~\ref{fig:rnp} we illustrate the implementation of this algorithm using a Stochastic Block Model (SBM) with four blocks of $40$ nodes, in which two nodes are linked with probability $p_{\mathrm{in}}=0.3$ if they belong to the same block and $p_{\mathrm{out}}\le 0.05$ otherwise.

The initialization is random rather than spectral. A Fiedler initialization would require global information and restrict rounds to bisections, whereas random initial conditions sample the multistable attractor landscape and often cut the graph into more than two components in one round. Since the dynamics on distinct components are decoupled, the recursion refines the partition component by component, without any agent knowing the number of communities.

\section{Theoretical Guarantees}\label{sec:theory}

For $\mathcal{W}\subseteq\Vcal$ and $i\in\mathcal{W}$, let $d_i^E=|\Ncal_i\setminus\mathcal{W}|$ be the number of neighbors of $i$ outside $\mathcal{W}$ and $\xi_i^{\mathcal{W}}=d_i^E/d_i$ its \emph{external degree ratio}.

\begin{figure}[t]
\centering
\begin{subfigure}[t]{0.32\columnwidth}
    \includegraphics[width=\linewidth]{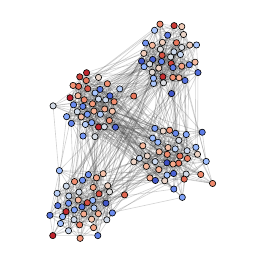}\label{fig:rnp1}
    \caption{Initialization.}
\end{subfigure}
\begin{subfigure}[t]{0.32\columnwidth}
    \includegraphics[width=\linewidth]{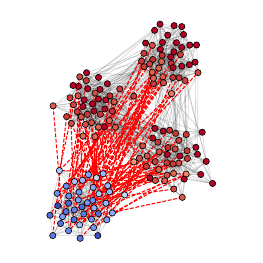}\label{fig:rnp1}
    \caption{Edge removal.}
\end{subfigure}
\begin{subfigure}[t]{0.32\columnwidth}
    \includegraphics[width=\linewidth]{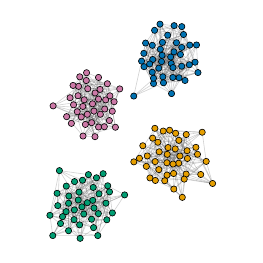}\label{fig:rnp1}
    \caption{Final partition.}
\end{subfigure}
\caption{RNP on an SBM with $p_{\mathrm{in}}=0.3$ and $p_{\mathrm{out}}\le 0.05$ ; $s=s_K$, $K=5$, $h=0.1$. (b) At steady state, edges joining agents with opposite actions (red) are cut. (c) Recursion yields four communities.}
\label{fig:rnp}\vspace{-3mm}
\end{figure}

\subsection{Spectral threshold for the discrete-time dynamics}

Let $\one^{\perp_{\bm{D}}}=\{\bm{x}\in\R^N\mid\langle\one,\bm{x}\rangle_{\bm{D}}=0\}$ be the subspace of vectors $D$-orthogonal to the consensus direction $\one$. This subspace is spanned by $\bm{v}_1,\dots,\bm{v}_{N-1}$ and invariant under $\bm{P}$. For $\bm{x}\in\R^N$, let $\bar x=\one^\top\bm{D}\bm{x}/W$ be its degree-weighted average, where $W=\sum_i d_i$, and $\Pi\bm{x}=\bm{x}-\bar x\one$ its $D$-orthogonal projection onto $\one^{\perp_{\bm{D}}}$, i.e., its disagreement component. One has $\Pi\bm{P}=\bm{P}\Pi$ and, for all $\bm{x},\bm{y}\in\R^N$,
\begin{equation}\label{eq:polar}
    \langle\Pi\bm{x},\Pi\bm{y}\rangle_{\bm{D}}=\frac{1}{2W}\sum_{i,j=1}^N d_id_j(x_i-x_j)(y_i-y_j).
\end{equation}

\begin{proposition}\label{prop:threshold}
    Let Assumption~\ref{ass:s} hold and $h\in(0,1]$.
    \begin{enumerate}
        \item If $K\lambda_{N-1}<1$, every equilibrium of~\eqref{eq:dyn_i} is of the form $\alpha\one$ with $\alpha\in\{-1,0,1\}$.\label{prop:threshold:1}
        \item If $K \rho(\Gcal) < 1$, then for all $t\ge0$, $\|\Pi\bm{x}(t)\|_{\bm{D}}\le\big(1-h(1-K\rho(\Gcal))\big)^t\|\Pi\bm{x}(0)\|_{\bm{D}}$.\label{prop:threshold:2}
        \item If $K\lambda_{N-1}<1$, the consensus set $\mathrm{span}(\one)$ is globally attractive for $\dot{\bm{x}}=\bm{P}\bm{s}(\bm{x})-\bm{x}$.\label{prop:threshold:3}
    \end{enumerate}
\end{proposition}

\begin{proof}
\ref{prop:threshold:1} First, let us note that equilibria satisfy $\bm{x}=\bm{P}\bm{s}(\bm{x})$. To simplify the writing we will replace $\bm{s}(\bm{x})$ by $\bm{s}$ when no confusion arises. Since $\Pi\bm{P}=\bm{P}\Pi$, equilibria satisfy $\Pi\bm{x}=\bm{P}\Pi\bm{s}$. Furthermore, the Rayleigh quotient of $\bm{P}$ on $\one^{\perp_{\bm{D}}}$ gives $\langle\Pi\bm{s},\Pi\bm{x}\rangle_{\bm{D}}\le\lambda_{N-1}\|\Pi\bm{s}\|_{\bm{D}}^2$. By Assumption~\ref{ass:s} one has 
\[(s(a)-s(b))(a-b)\ge\frac1K(s(a)-s(b))^2,\] 
so~\eqref{eq:polar} yields $\langle\Pi\bm{s},\Pi\bm{x} \rangle_{\bm{D}} \geq \|\Pi\bm{s}\|_{\bm{D}}^2/K$. Hence, 
\begin{equation*}
    (\frac{1}{K}-\lambda_{N-1})\|\Pi\bm{s}\|_{\bm{D}}^2\le0.
\end{equation*}
Therefore, if $K\lambda_{N-1}<1$ one gets that $\|\Pi\bm{s}\|_{\bm{D}}^2\le0$ yielding $\bm{s}=\alpha\one$ for some $\alpha\in\R$, $\bm{x}=\bm{P}\bm{s}=\alpha\one$ and $\alpha=s(\alpha)$. Recalling that Assumption \ref{ass:s} imposes $\mathrm{Fix}(s)=\{-1,0,1\}$, i.e. $\alpha\in\{-1,0,1\}$, completes the proof. 

\ref{prop:threshold:2} It is straightforward that 
\begin{equation*}
    \Pi\bm{x}(t+1)=(1-h)\Pi\bm{x}(t)+h\bm{P}\Pi\bm{s}.
\end{equation*}
The $\bm{D}$-norm of $\bm{P}$ on $\one^{\perp_{\bm{D}}}$ is $\rho(\Gcal)$, and~\eqref{eq:polar} with $|s(x_i)-s(x_j)|\le K|x_i-x_j|$ gives $\|\Pi\bm{s}\|_{\bm{D}}\le K\|\Pi\bm{x}\|_{\bm{D}}$. Replacing this in the equality above yields the inequality.

\ref{prop:threshold:3} This result is already stated in \cite{couthuresConsensusRobustClustering2025c}.
\end{proof}

Items~\ref{prop:threshold:1} and~\ref{prop:threshold:3} involve $\lambda_{N-1}$ only; item~\ref{prop:threshold:2} also requires $K|\lambda_1|<1$, which only matters for graphs with a strong bipartite component. With $s(x)=x$ ($K=1$) item~\ref{prop:threshold:1} holds on every connected graph: linear averaging cannot reveal communities. Graphs without bottleneck, such as complete graphs ($\lambda_{N-1}\le0$), cannot be split for any $K$. Item~\ref{prop:threshold:3} matters because item~\ref{prop:threshold:1} alone does not rule out a cut. A nonlinear system can have all its equilibria at consensus and still have trajectories that never reach them, for example oscillations.
Its limitation is that RNP runs the discrete-time dynamics, where item~\ref{prop:threshold:3} does not apply as is. The gap only matters for nearly bipartite graphs with a large step $h$; for small $h$ the discrete dynamics follows the continuous one.

\begin{proposition}[Unstable directions]\label{prop:local}
Under the assumptions of Proposition~\ref{prop:threshold}, the linearization of~\eqref{eq:dyn_i} at $\bm{0}$ is \[\bm{J}=(1-h)\bm{I}+hK\bm{P},\] with eigenvalues $\nu_k=1+h(K\lambda_k-1)$ and eigenvectors $\bm{v}_k$, $k=1,\dots,N$. The origin is unstable along $\one$ ($\nu_N>1$). Moreover, the invariant subspace $\one^{\perp_{\bm{D}}}$ contains an unstable direction if and only if $K\lambda_{N-1}>1$, the dominant one being the Fiedler vector $\bm{v}_{N-1}$; and $\nu_N>\nu_{N-1}$.
\end{proposition}

The proof follows from the spectral decomposition of $\bm{P}$. For $K>K^*(\Gcal)$, disagreements are amplified along the Fiedler vector, whose sign pattern is the spectral cut~\cite{spielmanSpectralPartitioningWorks2007}. However, the consensus mode grows faster, so a round may end in $\pm\one$ without any cut: this is why RNP needs retries.

\subsection{Robustness and confinement of dense cores}

Let $\mathcal{W}\subseteq\Vcal$, $\sigma\in\{-1,1\}$, and define $L_\sigma\ge 0$ the Lipschitz constant of $s$ on an interval $\mathcal{I}_\sigma$ which defines a neighborhood of $\sigma$ (note that for a saturation function $s_K$ one has $L_\sigma=0$ on $\{x\mid\sigma x\ge1/K\}$). Let us also introduce\\
$\bullet$ $\eta=\min_{i\in\mathcal{W}}\big(1-L_\sigma(1-\xi_i^{\mathcal{W}})\big)$, $\gamma=\max_{i\in\mathcal{W}}\oldfrac{\xi_i^{\mathcal{W}}}{1-L_\sigma(1-\xi_i^{\mathcal{W}})}$\\
$\bullet$ distance of the group to the polarized state $\sigma$: $\delta(t)=\max_{i\in\mathcal{W}}|x_i(t)-\sigma|$\\
$\bullet$ disturbance from external neighbors: $u(t)=\max_{i\in\mathcal{W},j\in\Ncal_i\setminus\mathcal{W}}|s(x_j(t))-\sigma|$

The following input-to-state stability (ISS) property holds (see ~\cite{couthuresConsensusRobustClustering2025c}): if $\eta>0$ and $x_i(\tau)\in\mathcal{I}_\sigma$ for all $i\in\mathcal{W}$, $\tau\le t$, then
\begin{equation}\label{eq:iss}
\delta(t)\le(1-h\eta)^t\delta(0)+\gamma\max_{0\le\tau<t}u(\tau).
\end{equation}
The gain $\gamma$ is smaller for more internally connected agents, but~\eqref{eq:iss} does not prevent a sign change. The next result does.

\begin{proposition}[Confinement]\label{prop:trap}
Let $s=s_K$, $h\in(0,1]$, and $\mathcal{W}\subseteq\Vcal$ satisfy
\begin{equation}\label{eq:trap}
\xi_i^{\mathcal{W}}\le\frac{K-1}{2K}\qquad\forall i\in\mathcal{W}.
\end{equation}
Then $\Hcal^{+}(\mathcal{W})=\{\bm{x}\in\Xcal\mid  x_i\ge 1/K,\ \forall i\in\mathcal{W}\}$ and $\Hcal^{-}(\mathcal{W})=\{\bm{x}\in\Xcal\mid  x_i\le -1/K,\ \forall i\in\mathcal{W}\}$ are forward invariant. Moreover, without information on the external neighbors \eqref{eq:trap} cannot be weakened.
\end{proposition}

\begin{proof}
Let us consider $\bm{x}(t)\in\Hcal^+(\mathcal{W})$, $i\in\mathcal{W}$, and let $\bar{s}_i= \sum_{j\in\Ncal_i}s_K(x_j(t)) / d_i$ be the average action of its neighbors. Internal neighbors broadcast $1$ and external ones at least $-1$, so $\bar s_i\ge1-2\xi_i^{\mathcal{W}}\ge1/K$, and 
\[x_i(t+1)=(1-h)x_i(t)+h\bar s_i\ge\min(x_i(t),\bar s_i)\ge1/K.\]
If~\eqref{eq:trap} fails for $i$ and all its external neighbors broadcast $-1$, then $\bar s_i<1/K$ and, for $h=1$, $x_i$ leaves $\Hcal^+(\mathcal{W})$. A similar reasoning can be done for $\Hcal^-(\mathcal{W})$.
\end{proof}

This closed form specializes the general force balance of~\cite{couthuresConsensusRobustClustering2025c}. Since $(K-1)/2K<1/2$, no saturated dynamics can confine a group whose members have more external than internal neighbors.

\begin{figure*}[t]
\centering
\includegraphics[width=\textwidth]{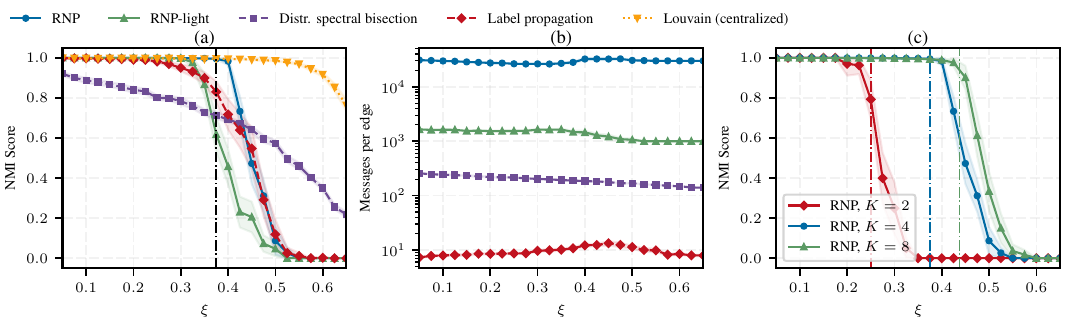}
\caption{ABCD graphs ($N=1000$, 20 graphs per $\xi$, 5 runs of each randomized method per graph; lines are means over the graphs, shaded areas 95\% confidence intervals). (a) NMI of RNP ($K=4$, $T_{\max}=300$, $N_{\mathrm{retry}}=50$), RNP-light ($T_{\max}=50$, $N_{\mathrm{retry}}=10$) and baselines; the vertical line is $(K-1)/2K=0.375$.  (b) Scalar messages per edge. (c) NMI of RNP for $K\in\{2,4,8\}$; vertical lines are the confinement thresholds $(K-1)/2K\in\{0.25,0.375,0.44\}$.}
\label{fig:bench}\vspace{-2mm}
\end{figure*}

\subsection{Exact separation in one round}

Let us notice that Proposition \ref{prop:trap} provides a condition for the preservation of a given community. The next result treats all groups simultaneously, provides a sufficient value for $T_{\max}$ and also gives a sufficient condition to reach polarization not just to stay there.
\begin{theorem}[One-round separation]\label{thm:recovery}
Let $s=s_K$, $h\in(0,1]$, and $\{\Vcal_1,\dots,\Vcal_m\}$ be a partition of $\Vcal$ into $m$ groups such that $\xi_i^{\Vcal_c}\le\varepsilon<\frac{(K-1)}{2K}$ for all $c\in\{1,\dots,m\}$ and $i\in\Vcal_c$. Let $m^*=\frac{\varepsilon}{((1-\varepsilon)K-1)}\in(0,\frac1K)$ and $\sigma=(\sigma_1,\dots,\sigma_m)\in\{-1,1\}^m$. 

If $\sigma_c x_i(0) > m^*$ for all $i \in \Vcal_c$ and all $c \in \{1,\dots,m\}$, then $\sigma_c x_i(t) \ge 1/K$ for all $i \in \Vcal_c$ and all $t\ge t^*$, where
\begin{equation*}
    t^*=\Big\lceil\frac{1/K-m_0}{h\big((1-\varepsilon)K-1\big)(m_0-m^*)}\Big\rceil,
\end{equation*}
with $m_0 = \min_{c, \, i \in \Vcal_c} \sigma_c x_i(0)$.
Consequently, if $T_{\max}\ge t^*$, the round cuts exactly the edges between groups with opposite signs $\sigma_c$, and no other edge.
\end{theorem}

\begin{proof}
Let $m(t)=\min_{c,\,i\in\Vcal_c}\sigma_cx_i(t)$. Since $s_K$ is odd and non-decreasing, for $i\in\Vcal_c$ the internal neighbors satisfy $\sigma_cs_K(x_j)\ge s_K(m)$ and the external ones $\sigma_cs_K(x_j)\ge-1$. So, with $\xi_i=\xi_i^{\Vcal_c}$ and $m=m(t)$ one gets 
\begin{equation*}
    \sigma_c x_i(t+1) \geq (1-h) m + h[(1 - \xi_i) s_K(m) - \xi_i].
\end{equation*}
The bracket is non-increasing in $\xi_i$ because $s_K(m)\ge-1$, hence $m(t+1)\ge g(m(t))$ with \[g(m)=(1-h)m+h[(1-\varepsilon)s_K(m)-\varepsilon],\] and $g$ is non-decreasing. On $[m^*,1/K]$ one has 
\begin{equation*}
    g(m) - m = h ((1 - \varepsilon) K - 1)(m - m^*) \geq 0,
\end{equation*}
so the iterates of $g$ from $m_0$ increase by at least $h((1-\varepsilon)K-1)(m_0-m^*)$ per step until they exceed $1/K$, which happens before $t^*$. For $m\ge1/K$, $g(m)\ge\min(m,1-2\varepsilon)\ge1/K$, since $1-2\varepsilon>1/K$. Thus $m(t)\ge1/K$ for $t\ge t^*$: every agent of $\Vcal_c$ broadcasts $\sigma_c$, and an edge is cut if and only if its endpoints lie in groups with opposite signs.
\end{proof}

Theorem~\ref{thm:recovery} is a deterministic recovery guarantee for one round. 
It thus certifies \emph{which} cut a round produces from a given set of initial conditions, not the typical behavior. Section~\ref{sec:num} shows that the actual basin is much larger. This is mainly due to the fact that the proposed uniform initialization covers a small set for large $N$.

\subsection{Consequences for RNP}

\begin{corollary}[Certificates for RNP]\label{cor:rnp}
    Let Assumption~\ref{ass:s} hold and $h\in(0,1]$.
    \begin{enumerate}[wide, nosep]
        \item \emph{Cost:} RNP stops after at most $(|\Ecal|+1)N_{\mathrm{retry}}$ rounds, and sends at most $2T_{\max}|\Ecal|$ scalar messages per round. At least $2T_{\max}N_{\mathrm{retry}}$ messages per edge are spent after the last cut.\label{cor:rnp:1}
        \item \emph{Core preservation:} Let $s=s_K$ and $\mathcal{W}$ satisfy~\eqref{eq:trap}. If the state enters $\Hcal^{+}(\mathcal{W})$ or $\Hcal^{-}(\mathcal{W})$ during a round, no edge inside $\mathcal{W}$ is cut in that round, and~\eqref{eq:trap} keeps holding in all later graphs where $\mathcal{W}$ is intact.\label{cor:rnp:2}
        \item \emph{Indivisibility:} If a component $\Gcal_c$ of the current graph satisfies $K\rho(\Gcal_c)<1$, every round on $\Gcal_c$ converges to the consensus set at the rate of Proposition~\ref{prop:threshold}; if the limit is $\pm\one$, no edge of $\Gcal_c$ is cut.\label{cor:rnp:3}
    \end{enumerate}
\end{corollary}
\begin{proof}
\ref{cor:rnp:1} Every round removes an edge or increments a counter that is reset only after a removal. \ref{cor:rnp:2} Proposition~\ref{prop:trap}, and pruning only decreases $d_i^E$. \ref{cor:rnp:3} Proposition~\ref{prop:threshold}-\ref{prop:threshold:2}, applied to $\Gcal_c$; components evolve independently.
\end{proof}

Items~\ref{cor:rnp:2}--\ref{cor:rnp:3} and Theorem~\ref{thm:recovery} suggest choosing $K$ in the window $K^*(\Gcal)<K<\min_cK^*(\Gcal_c)$, where $\Gcal_c$ is the subgraph induced by the $c$-th community to be found (with $\rho(\Gcal_c)$ in place of $\lambda_{N-1}$ for $\Gcal_c$ in discrete time): the graph can then be split, while isolated communities are certified indivisible. This window competes with the confinement threshold $(K-1)/2K$, which increases with $K$.

\section{Numerical Performance}\label{sec:num}

We use two benchmarks. The main one is the Artificial Benchmark for Community Detection (ABCD)~\cite{kaminskiArtificialBenchmarkCommunity2021}, whose communities, unlike those of the Lancichinetti--Fortunato--Radicchi (LFR) benchmark~\cite{lancichinettiBenchmarkGraphsTesting2008}, remain well connected as mixing increases. 
Its mixing parameter $\xi$ is the expected external degree ratio, i.e., the mean over nodes $i$ of $\xi_i^{\Vcal_c}$, where $\Vcal_c$ is the community of $i$. We generate graphs with $N=1000$, degree exponent $2.5$, degrees in $[10,100]$, community size exponent $1.5$ and sizes in $[20,250]$, using the Python implementation \texttt{abcd-graph}. Unless stated otherwise, RNP uses $s_K$ with $K=4$, $h=0.5$, $T_{\max}=300$, $N_{\mathrm{retry}}=50$. The second benchmark is LFR, generated with the \texttt{networkx} implementation and the same degree and community size parameters, with $8$ graphs per nominal mixing $\mu\in[0.05,0.65]$.

Baselines are Louvain~\cite{blondelFastUnfoldingCommunities2008} (centralized reference), asynchronous label propagation~\cite{raghavanLinearTimeAlgorithm2007} and a distributed spectral bisection (DSB) in the spirit of~\cite{kempeDecentralizedAlgorithmSpectral2008,linPowerIterationClustering2010}. In DSB, each component runs $30$ steps of the lazy walk $\bm{y}\gets\frac{(\bm{y}+\bm{P}\bm{y})}{2}$ from a random vector; the local differences $y_i(t)-y_i(t+1)$ cancel the consensus mode, and their signs approximate the Fiedler cut, which is applied as in RNP. A split is kept only if it increases the modularity; this global test is granted to DSB for free, which favors it. We report the normalized mutual information (NMI) between the detected and true partitions ($1$ for a perfect match, $0$ for independent partitions) and the number of scalar messages per edge.


\begin{figure}[t]
\centering
\includegraphics[width=\columnwidth]{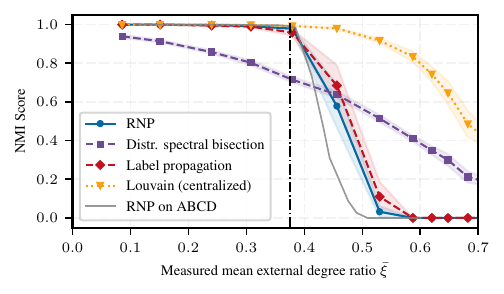}
\caption{Second benchmark: LFR graphs ($N=1000$, 20 graphs per $\mu$, 5 runs of each randomized method per graph), plotted against the measured mean external degree ratio $\bar\xi$. The gray line is RNP on ABCD (Fig.~\ref{fig:bench}a), also plotted against its measured $\bar\xi$; the vertical line is $(K-1)/(2K)=0.375$.}
\label{fig:lfr}\vspace{-4mm}
\end{figure}

\subsection{Results}

\emph{Accuracy.} RNP achieves NMI $\ge0.99$ and finds the correct communities (relative error below $0.01$) for $\xi\le0.4$, similar to Louvain (Fig.~\ref{fig:bench}a). In this regime it outperforms label propagation and DSB (NMI $0.72$ and $0.69$ at $\xi=0.4$, respectively), and $\leq 0.92$ even for well-separated graphs). Beyond $\xi\approx0.40$, the NMI of RNP drops (about $0.47$ at $\xi=0.45$, $0.09$ at $0.5$, $0$ from $0.55$) whereas Louvain stays above $0.96$ up to $\xi=0.55$; DSB degrades more slowly and is the most accurate decentralized method for $\xi\geq 0.45$. When it fails, RNP returns a single community (relative error close to $-1$): it \emph{fails safely} rather than producing a noisy partition.

\emph{Communication cost.} RNP needs about $3\times10^4$ messages per edge (Fig.~\ref{fig:bench}b), two orders of magnitude more than DSB ($140$ to $260$) and three more than label propagation (about $10$). This number of message is due to the certification of indivisibility: by Corollary~\ref{cor:rnp}, the final $N_{\mathrm{retry}}$ unsuccessful rounds alone cost $2T_{\max}N_{\mathrm{retry}}=3\times10^4$ messages per edge. RNP-light reduces it twentyfold with same accuracy for $\xi \leq 0.3$, but start degrading earlier ($0.46$ at $\xi=0.4$). This comes from the parameter$N_{\mathrm{retry}}$ that trades messages for accuracy, and not from $T_{\max}$. Since rounds start from independent initial conditions, a component that a round cuts with probability $p$ is wrongly declared indivisible with probability $(1-p)^{N_{\mathrm{retry}}}$, and $p$ is small near the breakdown. With $N_{\mathrm{retry}}=10$ and $T_{\max}=300$, RNP returns the same partitions for $\xi \leq 0.275$ with three times fewer messages, but its NMI drops quickly as $\xi$ increases.

\emph{Breakdown follows the confinement threshold.} Fig.~\ref{fig:bench}c varies $K$. The breakdown point moves with $K$ as predicted by Proposition~\ref{prop:trap}: RNP keeps an NMI of at least $0.97$ up to $\xi=0.2$ for $K=2$ (threshold $0.25$), up to $0.4$ for $K=4$ (threshold $0.375$), and up to $0.425$ for $K=8$ (threshold $0.44$). On these graphs, $K^*(\Gcal)$ grows from $1.1$ to $1.8$ and $\min_cK^*(\Gcal_c)$, over the subgraphs $\Gcal_c$ induced by the true communities, decreases from $2.4$ to $1.76$ as $\xi$ goes from $0.1$ to $0.5$. Hence $K=4$ and $K=8$ lie outside the indivisibility window, yet no over-segmentation is observed (relative error at most $0.06$). The window is therefore conservative in practice, and confinement is the binding constraint: the best choice is a large gain.

\emph{Second benchmark.} The LFR generator benchmark is mixing more the communities with its nominal parameter $\mu$ (e.g., a measured mean external ratio $\bar\xi=0.38$ for $\mu=0.25$ and $0.46$ for $\mu=0.3$). Simplify exposition, Fig.~\ref{fig:lfr} uses $\bar{\xi}$ as x-axis. RNP is again exact up to $\bar{\xi}\approx0.38$ (NMI $\ge 0.97$), drops to $0.58$ at $\bar\xi=0.46$ and to $0.03$ beyond $0.53$, i.e., it breaks down at the same external ratio as on ABCD, close to the confinement threshold. Below the breakdown, RNP is slightly more accurate than label propagation and clearly more accurate than DSB; at the breakdown, label propagation degrades slightly later, and DSB is the most accurate decentralized method beyond it. Louvain remains above $0.9$ up to $\bar\xi\approx0.53$. On these graphs, $K^*(\Gcal)$ exceeds $\min_cK^*(\Gcal_c)$ from $\mu\approx0.33$ on, and $6\%$ of the communities are disconnected at $\mu=0.4$, which confirms that LFR communities lose their internal cohesion as mixing increases.

\section{Limitations and Discussion}\label{sec:limits}

We can observe the following limitation for RNP:

\emph{Irreversible cuts.} A cut is decided from a single trajectory and never undone. Proposition~\ref{prop:trap} protects a core only once it has polarized, and Theorem~\ref{thm:recovery} covers an exponentially small set of initial conditions; frontier nodes are then assigned to one side or isolated depending on the initialization.

\emph{Cost of stopping.} Because the consensus mode dominates (Proposition~\ref{prop:local}), indivisibility is only certified statistically, after $N_{\mathrm{retry}}$ unsuccessful rounds. The retry counter also requires a component-wide agreement on whether a cut occurred; running a fixed number of rounds on a common clock is a fully decentralized alternative.

\emph{Gain selection.} The relevant thresholds depend on $K^*(\Gcal)$, $K^*(\Gcal_c)$ and the external ratios, which agents cannot compute. Our results support a large initial gain, but no local rule adapting $K$ is available yet.

\emph{Disclosure and scope.} Saturation limits but does not prevent information leakage, and no formal privacy guarantee is claimed. Proposition~\ref{prop:threshold} requires undirected graphs, and item~\ref{prop:threshold:2} the stronger condition $K\rho(\Gcal)<1$.

These limitations motivate two variants studied in the extended version: cutting edges during the transient with a decaying confidence bound, in the spirit of~\cite{morarescuOpinionDynamicsDecaying2011}, and replacing hard cuts by edge reliability scores over many rounds.

\section{Conclusion}\label{sec:conclusion}

We analyzed Recursive Neighbor Pruning, a decentralized community detection algorithm in which agents only exchange saturated actions and cut ties with neighbors displaying opposite actions. For the discrete-time dynamics run by the agents, we proved a spectral threshold for splitting, identified the Fiedler direction of symmetry breaking, gave a closed-form confinement condition $\xi_i^{\mathcal{W}}\le(K-1)/2K$ on the external degree ratios and an exact one-round separation result, and bounded the cost of the recursion. On ABCD and LFR benchmarks, RNP matches Louvain up to a mean external ratio of about $0.4$, its breakdown follows the confinement threshold on both benchmarks and for three gains, and it fails safely beyond. Its main weaknesses are its communication cost, dominated by stopping, and the lack of guarantees on the typical basin of the separating rounds, which future work will address together with local gain adaptation.

\bibliographystyle{IEEEtran}
\bibliography{Community_detection}

\end{document}